\documentclass[12pt]{article}
\usepackage{amsmath}
\usepackage{graphicx,psfrag,epsf}
\usepackage{enumerate}
\usepackage{psfrag,epsf}
\usepackage{url} 
\usepackage{amsmath,amssymb,amsfonts,amsthm,mathtools}
\usepackage{bm,bbm}
\usepackage{color}
\usepackage{subfigure}
\usepackage{algorithm,algpseudocode}
\usepackage{scalefnt}
\usepackage{multirow}
\usepackage{authblk} 
\usepackage[colorlinks=true, citecolor=blue, urlcolor=blue]{hyperref}
\usepackage{natbib}
\usepackage{dsfont}
\usepackage[title]{appendix}
\input cyracc.def

\usepackage[left=1.3in,top=1.25in,bottom=1.7in,right=0.85in]{geometry}

\usepackage{letltxmacro}
\makeatletter
\let\oldr@@t\r@@t
\def\r@@t#1#2{%
	\setbox0=\hbox{$\oldr@@t#1{#2\,}$}\dimen0=\ht0
	\advance\dimen0-0.2\ht0
	\setbox2=\hbox{\vrule height\ht0 depth -\dimen0}%
	{\box0\lower0.4pt\box2}}
\LetLtxMacro{\oldsqrt}{\sqrt}
\renewcommand*{\sqrt}[2][\ ]{\oldsqrt[#1]{#2}}
\makeatother
\theoremstyle{definition}

\newtheorem{theorem}{Theorem}[section]

\newtheorem{corollary}[theorem]{Corollary}
\newtheorem{example}{Example}[section]

\newtheorem{proposition}[theorem]{Proposition}
\newtheorem{remark}[theorem]{Remark}
\numberwithin{equation}{section} 

\makeatletter
\def\@seccntformat#1{\@ifundefined{#1@cntformat}%
	{\csname the#1\endcsname\quad}
	{\csname #1@cntformat\endcsname}
}
\makeatother

\newif\ifShowComments
\ShowCommentstrue
\def\strutdepth{\dp\strutbox}
\def\druk#1{\strut\vadjust{\kern-\strutdepth
        {\vtop to \strutdepth{%
                \baselineskip\strutdepth\vss
                        \llap{\hbox{#1}\quad}\null}}}}

\title{\bf
A closed-form expression for the variance of truncated distribution and its uses
}

\author[1, 2]{\hspace{-0.07cm}Roberto Vila \thanks{rovig161@gmail.com}}
\author[2]{\hspace{-0.09cm}Narayanaswamy Balakrishnan  \thanks{bala@mcmaster.ca}}
\author[1]{\hspace{-0.09cm}Raul Matsushita \thanks{raulmta@unb.br}}
\affil[1]{Department of Statistics, University of
	Bras\'ilia, Bras\'ilia, Brazil}
\affil[2]{
	Department of Mathematics and Statistics, McMaster University, Hamilton, Ontario, Canada}
\begin{document}
\maketitle
\begin{abstract}
This work sheds some light on the relationship between a distribution's standard deviation and its range, a topic that has been discussed extensively in the literature. While many previous studies have proposed inequalities or relationships that depend on the shape of the population distribution, the approach here is built on a family of bounded probability distributions based on skewing functions. We offer closed-form expressions for its moments and the asymptotic behavior as the support's semi-range tends to zero and $\infty$. We also establish an inequality in which the well-known Popoviciu's one is a special case. Finally, we provide an example using US dollar prices in four different currencies traded on foreign exchange markets to illustrate the results developed here.
\end{abstract}


\section{Introduction}

Relating the standard deviation ($\sigma$) to the range is a well-studied topic \citep[\emph{see}][]{Tippett, Popoviciu1935, Shone, David}. However, most of the results found in the literature in this regard propose inequalities or relationships that depend on the shape of the population distribution. \cite{Matsushita2020, Matsushita2023} recently suggested a power law between $\sigma$ and the semi-range $\ell$ without knowing the population distribution, but assuming  symmetric truncated forms restricted to $\ell \ll 1$. They argued that truncation is a phenomenon naturally generated by the sampling process. In their approach, conditional distribution properties can link the usual unbounded distribution for describing unobserved data.

Here, we establish a closed-form expression for the variance of truncated distributions valid for all $\ell > 0$. We start with the general case (Section \ref{sec:2.1}), by considering a truncated variable $X$ over the interval $[a,b]$ based on a cumulative distribution function $G$ with unbounded support as a skewing function. For $c \in (a,b)$ and $p>0$, we present a general form for $\mathbb{E}[(X-c)^p]$ as well as its asymptotic behavior as the support's semi-range tends to zero and $\infty$. We also present some inequalities that included Popoviciu's inequality as a special case. 

Section \ref{sec:2.2} presents some properties and examples regarding the symmetrically truncated distribution as a particular case. Importantly, we deduce the form of the ratio $\sigma/\ell$ as a function of $G$ and $\ell$. For illustrative purposes, we have presented an example using actual financial data (Section \ref{sec:3}). They consist of 16 million tick-by-tick returns of four currencies against the US dollar transacted on foreign exchange markets. Finally, Section \ref{sec:4} makes some brief concluding remarks.
 
\section{Main results}

\subsection{General case}
\label{sec:2.1}
Suppose we have a random variable with cumulative distribution function (CDF) $G(x)$ and with infinite support. Based on $G$, given two real numbers $a$ and $b$, such that $a<b$, we have
\begin{align}\label{def-F-p}
F_X(x)
=
{G(x)-G(a)\over G(b)-G(a)}, \quad a<x<b,
\end{align}
to be the truncated CDF of a random variable $X$ with support $(a,b)$. Then, we have the following result concerning its moments.

\begin{theorem}\label{lemma-initial}
Let $c$ and $p$ be real numbers such that $a<c<b$ and $p>0$.
Then, the $p$-th moment about $c$ is given by
\begin{align*}
\mathbb{E}[(X-c)^p]
=
\dfrac{1}{G(b)-G(a)}\,
	\big[
	(b-c)^p G(b)
	-
	(a-c)^p G(a)
-
p I_G(c;a-c,b-c,p)
\big],
\end{align*}
where $I_G(c;s,t,p)$ is defined as
\begin{align*}
	I_G(c; s,t,p)=\int_{s}^{t} y^{p-1} G(y+c)\,  {\rm d}y, \quad s<t.
\end{align*}
\end{theorem}
\begin{proof}
From \eqref{def-F-p}, we have
\begin{align*}
\mathbb{E}[(X-c)^p]
&=
{1\over G(b)-G(a)}
\int_{a}^{b} (x-c)^p\, {\rm d}G(x)
\\[0,2cm]
&
=
{p\over G(b)-G(a)}
\left\{
\int_{a}^{c} \left[\,\int_{0}^{x-c} y^{p-1}\, {\rm d}y\right] {\rm d}G(x)
+
\int_{c}^{b} \left[\,\int_{0}^{x-c} y^{p-1}\, {\rm d}y\right] {\rm d}G(x)
\right\}.
\end{align*}
Change the order of integration,
the above expression can be written as
\begin{align*}
&=
{p\over G(b)-G(a)}
\left\{
\int_{0}^{a-c}y^{p-1} \left[\,\int_{a}^{y+c} {\rm d}G(x) \right]  {\rm d}y
+
\int_{0}^{b-c}y^{p-1} \left[\,\int_{y+c}^{b} {\rm d}G(x) \right]  {\rm d}y
\right\}
\nonumber
\\[0,2cm]
&=
{1\over G(b)-G(a)}
\left[
{(b-c)^p G(b)}
-
{(a-c)^p G(a)}
\right] \nonumber
\\[0,2cm]
&
-
{p\over G(b)-G(a)}
\int_{a-c}^{0} y^{p-1} G(y+c)\,  {\rm d}y
-
{p\over G(b)-G(a)}
\int_{0}^{b-c} y^{p-1} G(y+c)\,  {\rm d}y,
\end{align*}
which completes the proof of the theorem.
\end{proof}

\begin{corollary}\label{prop-bound-partial}
	Under the conditions of Theorem \ref{lemma-initial}, 
	we have
	\begin{align*}
	\mathbb{E}[(X-c)^p]
	\leqslant
		\begin{cases}
		(b-c)^p [1-F_X(c)]
		+
		(a-c)^{p} F_X(c), & \text{if} \ p \ \text{is even},
		\\[0,2cm]
		(b-c)^p [1-F_X(c)], & \text{if} \ p \ \text{is odd},
		\end{cases}
	\end{align*}
	where $F_X$ is as given in \eqref{def-F-p}.
\end{corollary}
\begin{proof}
It is evident that
\begin{align}\label{decomposition of I}
I_G(c;a-c,b-c,p)=I_G(c;a-c,0,p)+I_G(c;0,b-c,p).
\end{align}

Assuming $p$ to be even, for $a-c<y<0$, we have $y^{p-1} G(c)< y^{p-1} G(y+c)< y^{p-1}G(a)$, and for $0<y<b-c$, we have $y^{p-1}G(c)<y^{p-1} G(y+c)<y^{p-1}G(b)$. Consequently, we get
\begin{align*}
I_G(c;a-c,0,p)\geqslant -{G(c) (c-a)^p\over p} 
\quad 
\text{and} 
\quad
I_G(c;0,b-c,p)\geqslant {G(c) (b-c)^p\over p}.
\end{align*}
Now, upon using these inequalities in \eqref{decomposition of I}, we obtain
\begin{align*}
	I_G(c;a-c,b-c,p)\geqslant  {G(c)\over p}\, \big[(b-c)^p-(c-a)^p\big].
\end{align*}
Using the above inequality in the expression in Theorem \ref{lemma-initial}, we get
\begin{align*}
\mathbb{E}[(X-c)^p]
\leqslant
\dfrac{1}{G(b)-G(a)}\,
\big\{
(b-c)^p G(b)
-
(a-c)^p G(a)
-
G(c) \big[(b-c)^p-(c-a)^p\big]
\big\}.
\end{align*}
From \eqref{def-F-p}, the right-hand side of the above inequality can be rewritten as 
\begin{align*}
=
(b-c)^p [1-F_X(c)]
+
(a-c)^{p} F_X(c). 
\end{align*}
This proves the inequality for the case when $p$ is even.

The inequality for the case when $p$ is odd can be established in an analogous manner.
\end{proof}

When $p$ is even, a lower bound for $\mathbb{E}[(X-c)^p]$ can be established as below.
\begin{proposition}\label{prop-lower-bound}
Under the conditions of Theorem \ref{lemma-initial}, 
for $p$ even, we have
\begin{align*}
	\mathbb{E}[(X-c)^p]
	\geqslant 
	(t-c)^p \left[F_X(b)-F_X(t)\right], \quad \text{if} \ c<t<b.
\end{align*}
\end{proposition}
\begin{proof}
Suppose $p$ is even. Then, as
$(X-c)^p\geqslant 0$ and $(X-c)^p\geqslant (t-c)^p$ for $X>t>c$, it is clear that
\begin{align*}
\mathbb{E}[(X-c)^p]
\geqslant
\mathbb{E}\big[(X-c)^p\mathds{1}_{\{X>t\}}\big]
\geqslant
(t-c)^p\,
\mathbb{E}\big[\mathds{1}_{\{X>t\}}\big],
\end{align*}
which yields the required result.
\end{proof}

\begin{proposition}\label{popoviciu-generalized}
	Under the conditions of Theorem \ref{lemma-initial}, 
	we have
\begin{align*}
	\min_{a<c<b}\mathbb{E}[(X-c)^p]
	\leqslant
	\begin{dcases}
	\left(\dfrac{b-a}{2}\right)^p, & \text{if} \ p \ \text{is even},
	\\[0,2cm]
	0, & \text{if} \ p\geqslant 1 \ \text{is odd}.
	\end{dcases}
\end{align*}
\end{proposition}
\begin{proof}
Suppose $p$ is even. In this case, from Corollary \ref{prop-bound-partial}, we have
\begin{align}\label{des-1}
\mathbb{E}[(X-c)^p]
\leqslant
(b-c)^p [1-F_X(c)]
+
(a-c)^{p} F_X(c)
\leqslant
S(F(c))
T(c),
\end{align}
where $S(F(c))=\max\left\{1-F_X(c), F_X(c)\right\}$ and $T(c)=(b-c)^p
+
(a-c)^{p}$. As $F_X(c)=1/2$ is a minimum point of $M(F(c))$, we have
\begin{align*}
	\min_{0<F(c)<1} S(F(c))={1\over 2}.
\end{align*}
Taking the minimum over $0<F(c)<1$ in \eqref{des-1}, we get
\begin{align*}
\mathbb{E}[(X-c)^p]
\leqslant
{1\over 2}\,
T(c).
\end{align*}
Now, taking the minimum over $a<c<b$ in the above inequality and using the fact that the function $T(c) $ reaches a minimum value at the point $c=(a+b)/2$, we get
\begin{align*}
	\min_{a<c<b} \mathbb{E}[(X-c)^p]\leqslant {1\over 2}\, T\left({a+b\over 2}\right)=\left(\dfrac{b-a}{2}\right)^p.
\end{align*}
This proves the inequality for the case when $p$ is even. Further, 
the inequality for the case when $p$ is odd trivially follows from Corollary \ref{prop-bound-partial}.
%
%
%
\end{proof}

\begin{proposition}\label{Main-Theorem}
The variance $\sigma^2$ can be expressed as
	\begin{align*}
	\sigma^2 = {\rm Var}(X)
		=
		\dfrac{1}{G(b)-G(a)}\,
		\big[
		(b-\mu)^2 G(b)
		-
		(a-\mu)^2 G(a)
		-
		2 I_G(\mu;a-\mu,b-\mu)
		\big],
	\end{align*}
	where $\mu=\mathbb{E}(X)$ and  
	\begin{align*}
	I_G(\mu;s,t)=\int_{s}^{t} y G(y+\mu)\,  {\rm d}y, \quad s<t.
	\end{align*}
\end{proposition}
\begin{proof}
By taking $c=\mu$ and $p=2$ in Theorem \ref{lemma-initial}, 
we obtain the required result.
\end{proof}

\begin{proposition}
The Popoviciu inequality \citep{Popoviciu1935} on variances given by
\begin{align*}
	\sigma^2\leqslant\left(\dfrac{b-a}{2}\right)^2,
\end{align*}
follows from Proposition \ref{popoviciu-generalized}.
\end{proposition}
\begin{proof}
	The proof follows immediately by setting $p=2$ in  Proposition \ref{popoviciu-generalized}.
\end{proof}

A reverse form of Popoviciu's inequality can be obtained upon taking $p=2, c=\mu=\mathbb{E}(X)$ and $t=(\mu+b)/2$ in Proposition \ref{prop-lower-bound}.
\begin{proposition}
We have
\begin{align*}
	\sigma^2 \geqslant\left(\dfrac{b-\mu}{2}\right)^2 \left[F_X(b)-F_X\left({\mu+b\over 2}\right)\right].
\end{align*}
\end{proposition}

\subsubsection*{Asymptotic behavior}
In the following theorem, we establish the asymptotic behaviour of the $p$-th moment.

\begin{theorem}\label{prop-limits}
	If $X$ is distributed as in \eqref{def-F-p}, then
	we have
\begin{align*}
\lim_{{b-a\over 2}\longrightarrow 0^+} \mathbb{E}\left[\left(\dfrac{X-a}{b-a}\right)^p\,\right]
=
{1\over p+1}, \quad p>-1,
\end{align*}
and
\begin{align*}
\lim_{{b-a\over 2}\longrightarrow \infty} \mathbb{E}\left[\left(\dfrac{X-a}{b-a}\right)^p\,\right]
=
{1\over 2^p}.
\end{align*}
\end{theorem}
\begin{proof}
From \eqref{def-F-p}, we find
	\begin{align*}
	F_{\frac{X-a}{b-a}}(z)
	=
	{G\big(z(b-a)+a\big)-G(a)\over G(b)-G(a)}
	=
	\dfrac{G\big({b-a\over 2}\,(2z-1)+{a+b\over 2}\big)-G\big({a+b\over 2}-{b-a\over 2}\big)}{ 
	G\big({b-a\over 2}+{a+b\over 2}\big)-G\big({a+b\over 2}-{b-a\over 2}\big)}, \ 0<z<1.
	\end{align*}	
Now, it is a simple task to verify that
\begin{align*}
\lim_{{b-a\over 2}\longrightarrow 0} F_{\frac{X-a}{b-a}}(z)
=
F_U(z), 
\quad \forall z\in\mathbb{R},
\end{align*}
and
\begin{align*}
\lim_{{b-a\over 2}\longrightarrow \infty} F_{\frac{X-a}{b-a}}(z)
=
F_Y(z), 
\quad \forall z\neq{1\over 2},
\end{align*}
where $U\sim U(0,1)$ and $Y$ is a discrete random variable such that $\mathbb{P}(Y=1/2)=1$. 

Moreover, we note that $[(X-a)/(b-a)]^p$ is uniformly integrable because $0<(X-a)/(b-a)<1$. As convergence in distribution along with uniform integrability imply convergence in mean \cite[cf.][ Theorem 5.4]{Billingsley1968}, we have
\begin{align*}
\lim_{{b-a\over 2}\longrightarrow 0^+} \mathbb{E}\left[\left(\dfrac{X-a}{b-a}\right)^p\,\right]
=
\mathbb{E}(U^p), \quad p>-1,
\end{align*}
and
\begin{align*}
\lim_{{b-a\over 2}\longrightarrow \infty} \mathbb{E}\left[\left(\dfrac{X-a}{b-a}\right)^p\,\right]
=
\mathbb{E}(Y^p),
\end{align*}
which completes the proof of the theorem.
\end{proof}

\begin{corollary}\label{corollary-0}
We further have
	\begin{align*}
	\lim_{{b-a\over 2}\longrightarrow 0^+} {\sigma^2 +(\mu-a)^2\over (b-a)^2}={1\over 3}
	\quad
	\text{and}
	\quad 
	\lim_{{b-a\over 2}\longrightarrow \infty} {\sigma^2\over (b-a)^2}= 0.
	\end{align*}
where $\mu=\mathbb{E}(X)$ and $\sigma^2 = {\rm Var}(X)$.
\end{corollary}
\begin{proof}
	By taking $p=2$ in Theorem \ref{prop-limits}, the above results follow readily.
\end{proof}

\subsection{Symmetric case}
\label{sec:2.2}

	In this section, we assume that $G$ is a skewing function, i.e., it is such that $G(x)\geqslant 0$ and $G(-x)=1-G(x)$, and $X$ distributed as in \eqref{def-F-p} with $a=-\ell$ and $b=\ell>0$. That is, $X$ has its CDF as
\begin{align}\label{def-F-s}
F_X(x)
=
{G(x)+G(\ell)-1\over 2G(\ell)-1}, \quad -\ell<x<\ell. 
\end{align}

\begin{proposition}\label{prop-var-ell}
The variance $\sigma^2$ can be expressed as
	\begin{align*}
	\sigma^2
	=
	\ell^2 H(\ell),
	\end{align*}
where 
       \begin{align*}
       H(\ell) = 1-\dfrac{2C(\ell)-1}{2G(\ell)-1},
       \end{align*}
 with $C(\ell)=C(\ell,G)$ defined as
	\begin{align*}
	C(\ell)={2\over \ell^2} \int_0^\ell yG(y)\, {\rm d}y.
	\end{align*}
	Moreover, $1/2\leqslant C(\ell)\leqslant G(\ell)$.
\end{proposition}
\begin{proof}
As $G$ is a skewing function, we have $\mu=\mathbb{E}(X)=0$. Moreover, $I_G(\mu;a-\mu,b-\mu)$ in Proposition \ref{Main-Theorem} satisfies the identity 
	\begin{align*}
	I_G(\mu;a-\mu,b-\mu)=\ell^2\left[C(\ell)-{1\over 2}\right].
	\end{align*}
Upon substituting this in Proposition \ref{Main-Theorem} and carrying out some simple algebraic steps, the required result follows.
\end{proof}

\begin{remark}
	It is useful to observe that, knowing $C(\ell)$ (see Table \ref{table:1} for some explicit examples of these constants), Proposition \ref{prop-var-ell} gives a more informative result than Popoviciu's inequality and present
	in particular a method for the exact calculation of the variance of truncated distributions of the form in \eqref{def-F-s}.
\end{remark}
\begin{table}[H]
	\caption{Some examples of constants  $C(\ell)$, for use in Proposition  \ref{prop-var-ell}.
	}
	\vspace*{0.15cm}
	\centering 
	\resizebox{\linewidth}{!}{
		\begin{tabular}{llll} 
			\hline
			Distribution
			& $G(x)$ &  $C(\ell)$ 
			\\ [0.5ex] 
			\noalign{\hrule }
			Normal 
			& $\Phi(x)$ &  
			${1\over 2\ell^2} 
			\left\{
			\ell \left[\ell + \exp\big(-{\ell^2\over 2}\big) \sqrt{{2\over \pi}}\,\right] + (\ell^2-1) {\rm erf}\big({\ell\over\sqrt{2}}\big)
			\right\}$ 
			\\ [1ex] 
			Student-$t$ ($\nu=2$) 
			& ${1\over 2}\left(1+\sqrt{x^2\over x^2+1}\, \right)$ 
			&   
			${1\over 2\ell^2} 
			\left\{
			\ell^2 
			+ 
			{1\over\ell}
				\sqrt{\ell^2\over 2 + \ell^2} 
				\left[
				2\ell + \ell^3 - 2\sqrt{2 + \ell^2} {\rm arcsinh}({\ell\over \sqrt{2}})
				\right]
			\right\}$
			\\ [1ex]
			Cauchy 
			& ${1\over\pi}\arctan(x)+{1\over 2}$  &  
			${1\over 2\pi\ell^2}\big[\ell(\ell\pi-2) + 2(1 + \ell^2) \arctan(\ell)\big]$
			\\ [1ex]   
			Laplace 
			& ${1\over 2}+{1\over 2}{\rm sgn}(x)[1-\exp(-\vert x\vert)]$   & 
			${1\over \ell^2} \exp(-\ell) \big[1 + \ell + \exp(\ell) ( \ell^2-1)\big]$
					\\ [1ex]   
			Logistic 
					& ${1\over 1+\exp(-x)}$  & 
					${2\over\ell^2}
					\left\{
					{\pi^2\over 12} + \ell\log[1 + \exp(\ell)] + {\rm Li}_2[-\exp(\ell)]
					\right\}$
			\\[3ex]  	
			\hline	
		\end{tabular}
	}
	\label{table:1} 
\end{table}
\noindent
In Table \ref{table:1}, $\Phi$ is the CDF of a standard normal distribution,
${\rm Li}_2[z]=-\int_0^z \log(1-x)/x {\rm d}x$ is the polylogarithm function  of order $2$ and ${\rm erf}(z)=(2/\sqrt{\pi}) \int_0^z \exp(-x^2){\rm d}x$ is the error function.

\begin{proposition}\label{prop-1}
We further have
\begin{align*}
    \lim_{\ell\longrightarrow 0^+}H(\ell) = \lim_{\ell\longrightarrow 0^+} {\sigma^2\over \ell^2}={1\over 3}
\quad
\text{and}
\quad 
\lim_{\ell\longrightarrow\infty}H(\ell) = \lim_{\ell\longrightarrow \infty} {\sigma^2\over \ell^2}= 0.
\end{align*}
\end{proposition}
\begin{proof}
Upon taking $a=-\ell$, $b=\ell$ and $\mu=0$ in Corollary \ref{corollary-0}, the required result follows.
\end{proof}

Next, we present two further examples in addition to Table \ref{table:1}.

\begin{example}
\label{ex:TCauchy}
Let $X$ be a truncated symmetric standard Cauchy distribution with density function $f(x) = \left(2 \arctan{\ell} \right)^{-1} \cdot \left(1+x^2\right)^{-1}$ if $|x| < \ell$, with $ \ell \leqslant 1$, and $f(x)=0$, if $|y| > \ell$. As its variance is $\sigma^2 = \ell/\arctan\ell - 1$ \citep[see][]{Johnson.Kotz.1970}, we obtain
       \begin{align*}
       H(\ell) = \frac{1}{\ell \arctan\ell} - \frac{1}{\ell^2}.
       \end{align*}
\end{example}

\begin{example}
\label{ex:TNormal}
Let $X$ be a symmetrically truncated standard Gaussian distribution with density function $f(x) = \{\sqrt{2\pi} [1-2\Phi(-\ell)]\}^{-1}\exp({-x^2/2})$, if $|x| < \ell$, where $ \ell \leqslant 1$ and $\Phi$ is the standard Gaussian cumulative distribution function, and $f(x)=0$, if $|x| > \ell$ \citep[see][]{Johnson.Kotz.1970}. As its variance can be expressed as $\sigma^2 = 1 - 2 \ell f(\ell)$, we find
       \begin{align*}
       H(\ell) = \frac{1}{\ell^2} - \frac{2}{\ell^3 f(\ell)}.
       \end{align*}
\end{example}

\noindent 
In both these examples, we observe $H(\ell) \rightarrow 1/3$ as $\ell \downarrow 0$ and $H(\ell) \uparrow 0$ as $\ell \uparrow \infty$, as stated in Proposition \ref{prop-1}. However, 
       \begin{align*}
       \ell H(\ell) = \frac{\sigma^2}{\ell}
       \end{align*}
behaves differently as $\uparrow \infty$. In the first example,  $\ell H(\ell)  \rightarrow  2/\pi$, but in the second example, we find  $\ell H(\ell)  \rightarrow  0$. 

\section{Illustration with financial data} 
\label{sec:3}
\noindent We illustrate the results developed here with intraday spot exchange rate data of four currencies against the US dollar transacted on the foreign exchange (Forex) market. There are 16 million tick-by-tick returns of bid prices provided by Tick Data, LLC (Table \ref{table:2}). Following the discussion regarding the truncated nature of the past \citep{Matsushita2020, Matsushita2023}, we consider the symmetric case here. 

For each currency, let $\{X_{ij}\}$ be the $j$th return observed on day $i = 1,\ldots, d$ (Table \ref{table:2}). Taking the daily sample standard deviation $s_i$ and the maximum daily absolute return $\ell_i = \max_{1\leq j \leq n_i}\{X_{ij}\}$, where $n_i$ denotes the sample size on day $i$ with $n = \sum_{i=1}^{d} n_i$, Figure \ref{figure:1} depicts the daily sample ratios $\{s_i/\ell_i\}$ in the form of dots.

Now, consider the general sequence ignoring days as $\{X_{t}\}$, where $t = 1,\ldots, n$. Letting 
$\ell^* = \max_{1\leq t \leq n}\{X_{t}\}$, we generated a grid of 1,000 truncation points, 
$\{\ell: \ell = m \ell^*/1000, \text{where } m = 1,\ldots, 1000\}$. For each $\ell$ over this grid, we obtained the sample standard deviation of the conditional (truncated) data $\{X_{t}| |X_{t}| \leq \ell\}$.
In this way, we empirically find the form of the ratio $\sigma/ \ell$ for the returns of each currency (Figure \ref{figure:1}). Then, Proposition \ref{prop-var-ell} provides a feasible and practical way of describing the relationship between the variance and the cutoff $\ell$. For small $\ell$, \cite{Matsushita2020, Matsushita2023} proposed the power law $\sigma^\beta/ \ell \approx \zeta$ from a second-order approximation, where $\beta > 0$ and $\zeta$ are real constants. So, we may approximate $\sigma/\ell$ as
\begin{align*}
\frac{\sigma}{\ell} \approx \zeta^{1/\beta} \ell^{- 1 + 1/\beta}.
\end{align*}
Figure \ref{figure:2} depicts the log-log plots of this approximated result, and shows the validity of such a power law approximation for $\sigma/\ell$.

\begin{table}[H]
\caption{Intraday spot exchange rate data description.}
\label{table:2}
\centering
\resizebox{\textwidth}{!}{
\begin{tabular}{l l l r r c c} 
 \hline
 Country & Currency & Code & \multicolumn{1}{c}{Period} & Number of days ($d$)& Data points ($n$) \rule{0pt}{2.6ex} \\ [0.5ex] 
 \hline
Britain        & British pound      & GBP  & 31 Aug 08 $-$ 12 Jun 15 & 2,116 & 2,754,615 \\
Canada         & Canadian dollar    & CAD  & 12 Jun 00 $-$ 12 Jun 15 & 4,419 & 3,931,202 \\ 
Japan          & Japanese yen       & JPY  & 30 May 00 $-$ 12 Jun 15 & 4,598 & 4,804,463 \\
Switzerland    & Swiss franc        & CHF  & 30 May 00 $-$ 12 Jun 15 & 4,587 & 4,838,100 \\
& & & Total&  15,720   &  16,328,380\\
\hline
Source: Tick Data, LLC. & & & &
\end{tabular}}
\end{table}

\begin{figure}[H]
\centering
\makebox{\includegraphics[width=6 cm, height=5 cm]{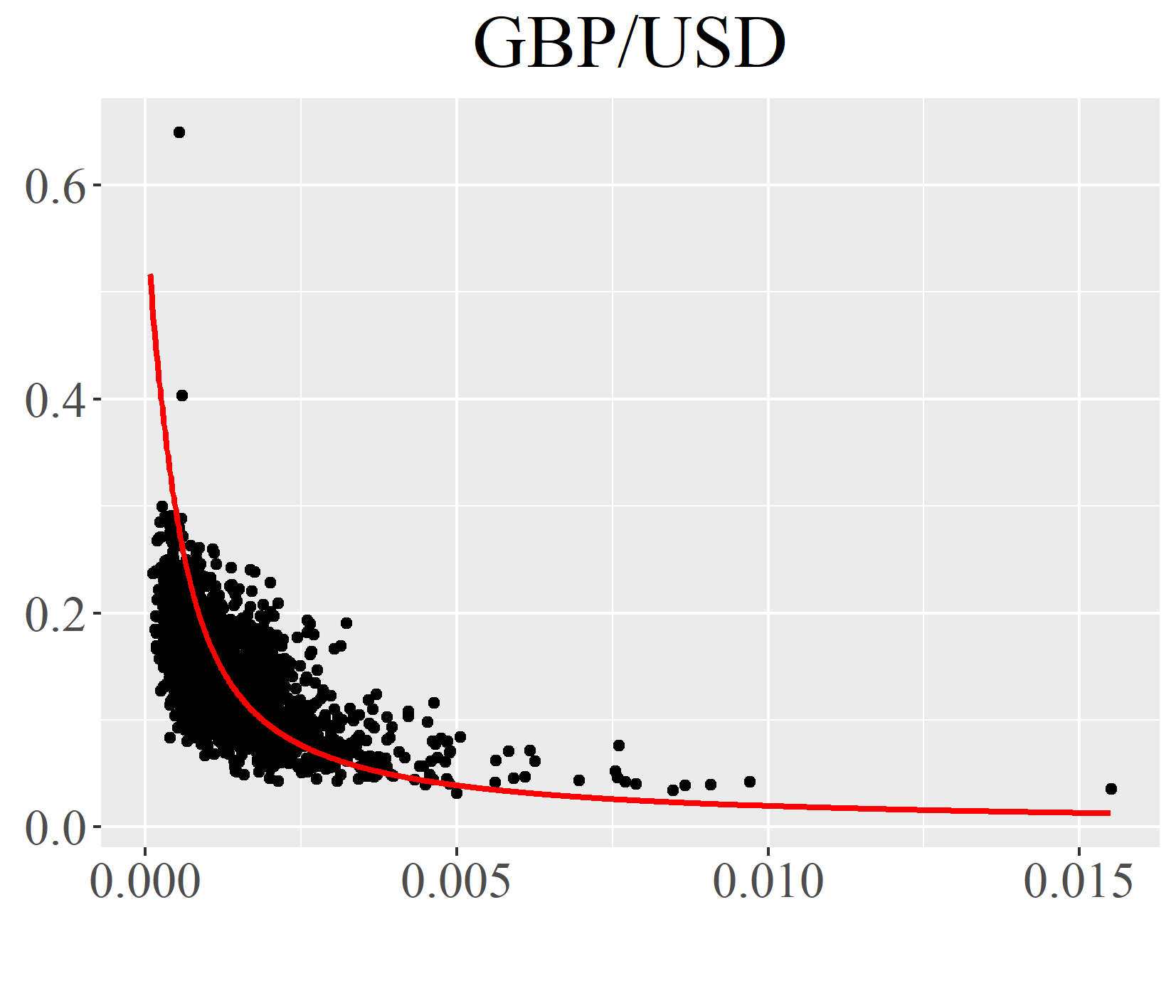}}
\makebox{\includegraphics[width=6 cm, height=5 cm]{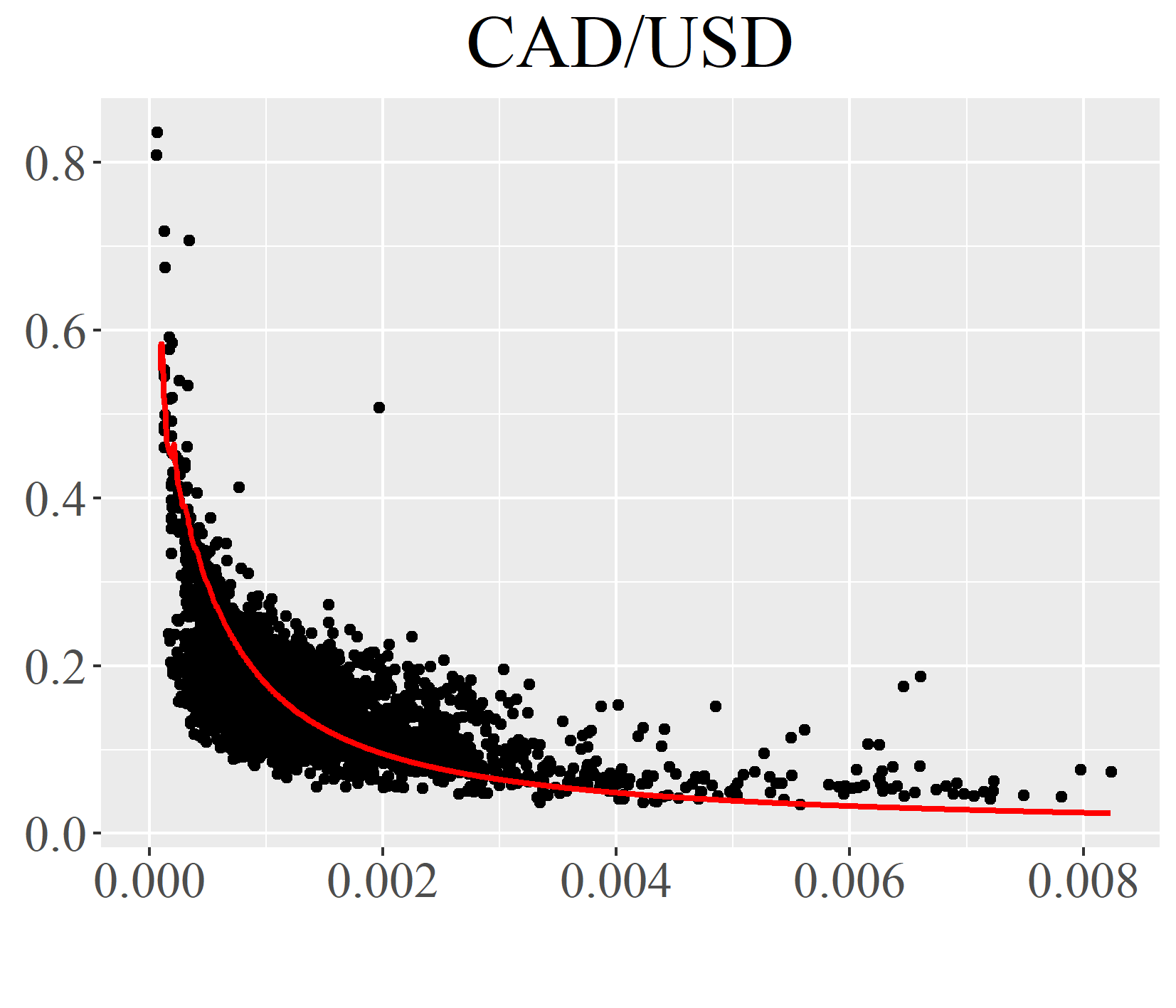}}\\
\makebox{\includegraphics[width=6 cm, height=5 cm]{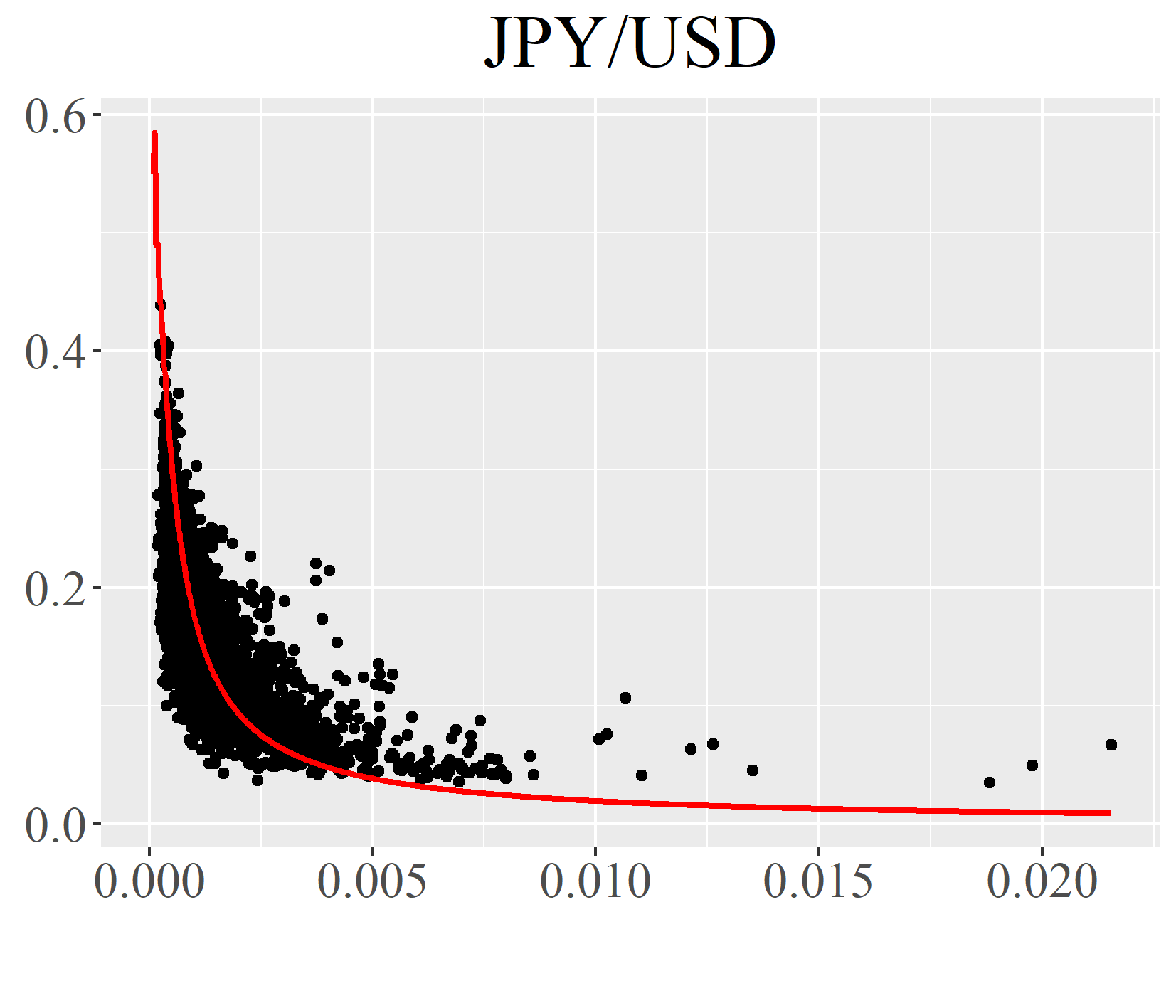}}
\makebox{\includegraphics[width=6 cm, height=5 cm]{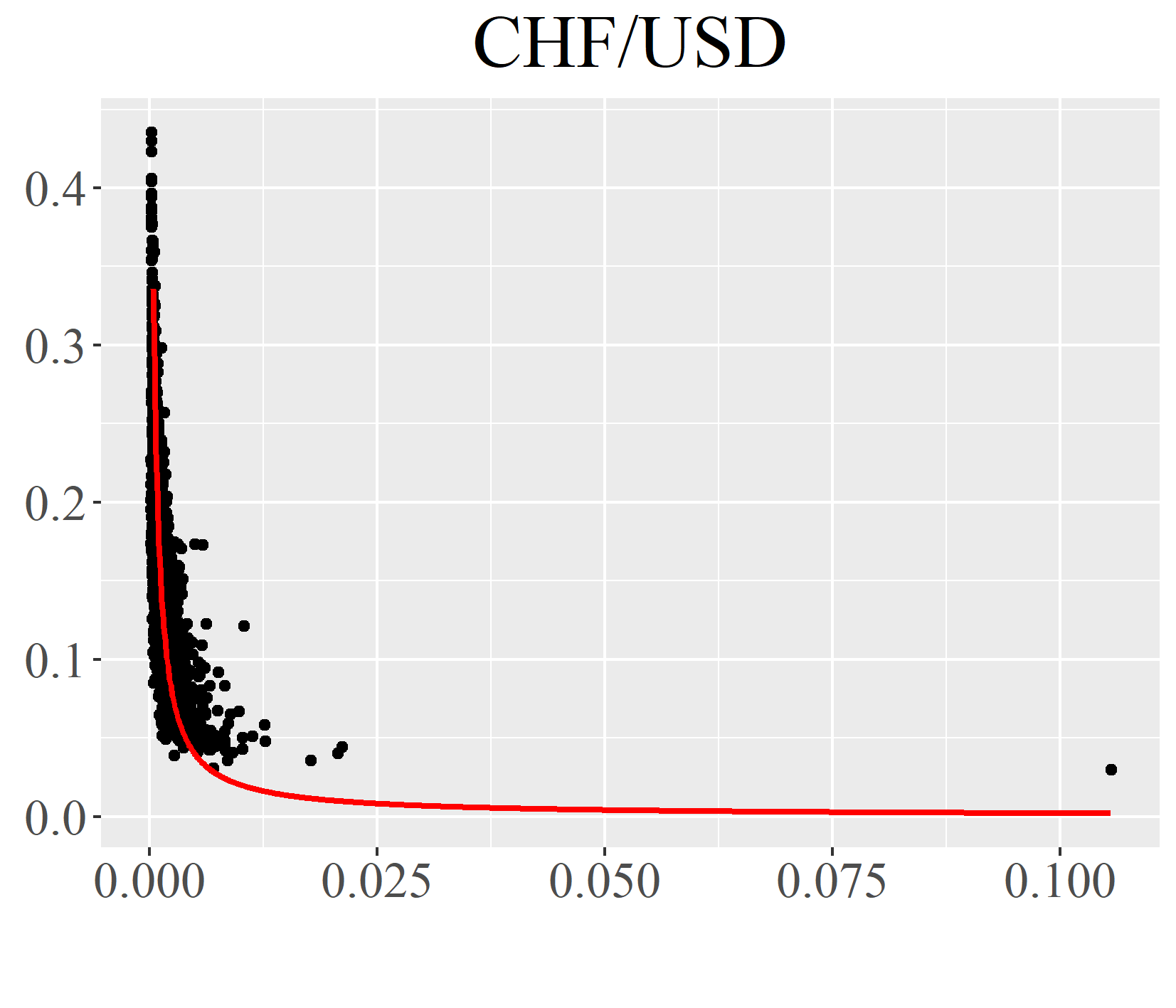}}
\caption{\label{figure:1}\footnotesize Display of $\ell$ \emph{versus} the sample ratio $\sigma/ \ell$ (lines) from data described in Table \ref{table:2}, where dots represent daily values.}
\end{figure}

\begin{figure}[H]
\centering
\makebox{\includegraphics[width=6 cm, height=5 cm]{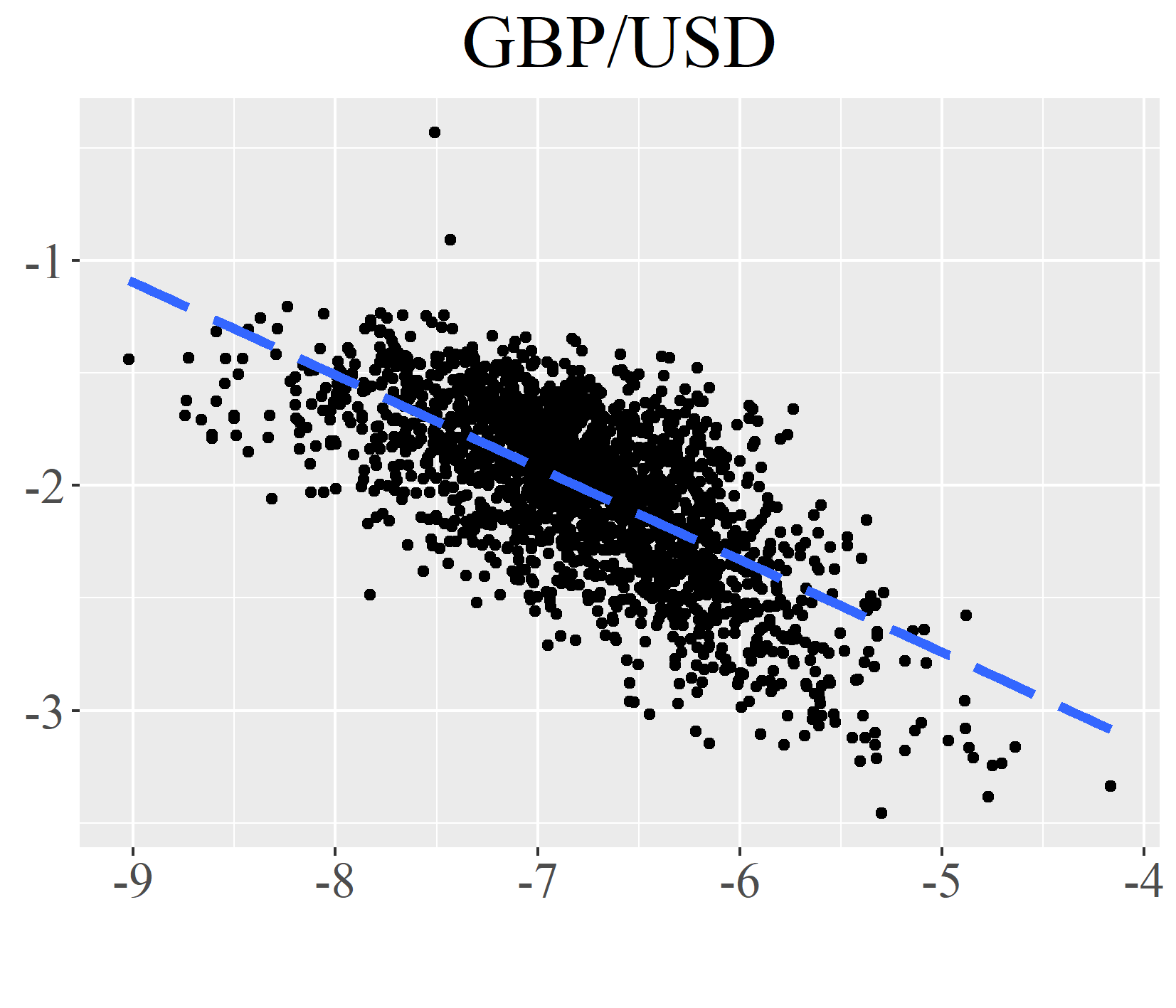}}
\makebox{\includegraphics[width=6 cm, height=5 cm]{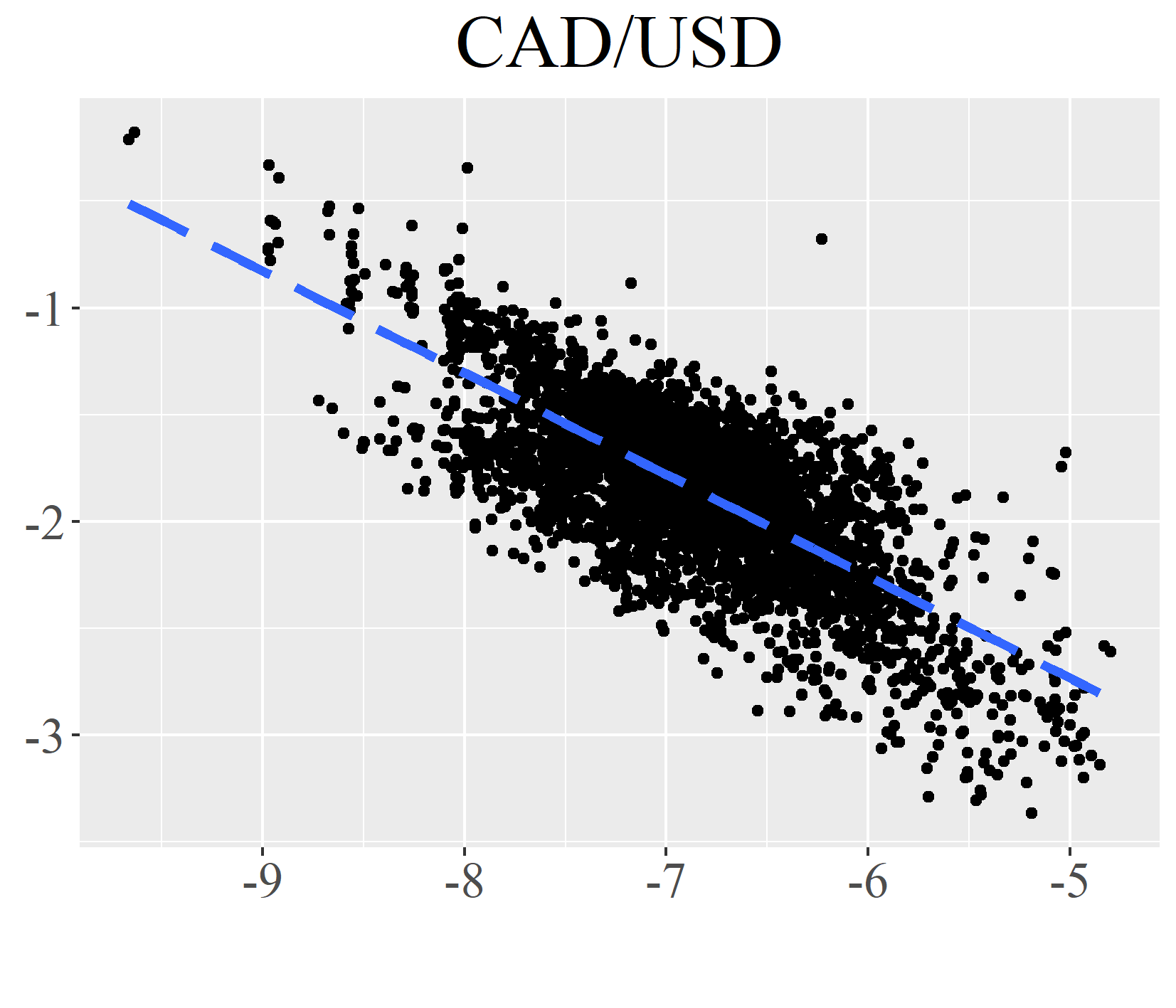}}\\
\makebox{\includegraphics[width=6 cm, height=5 cm]{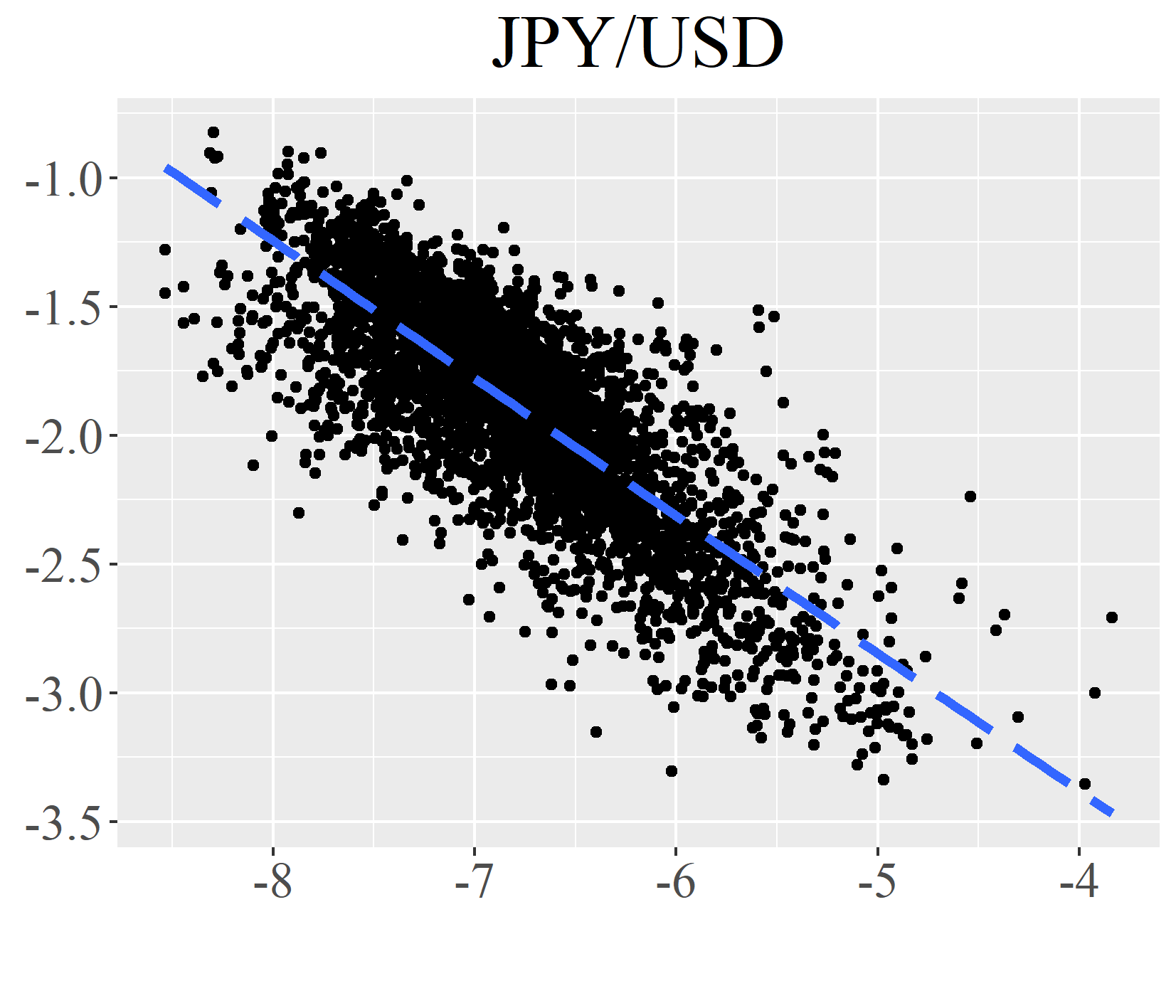}}
\makebox{\includegraphics[width=6 cm, height=5 cm]{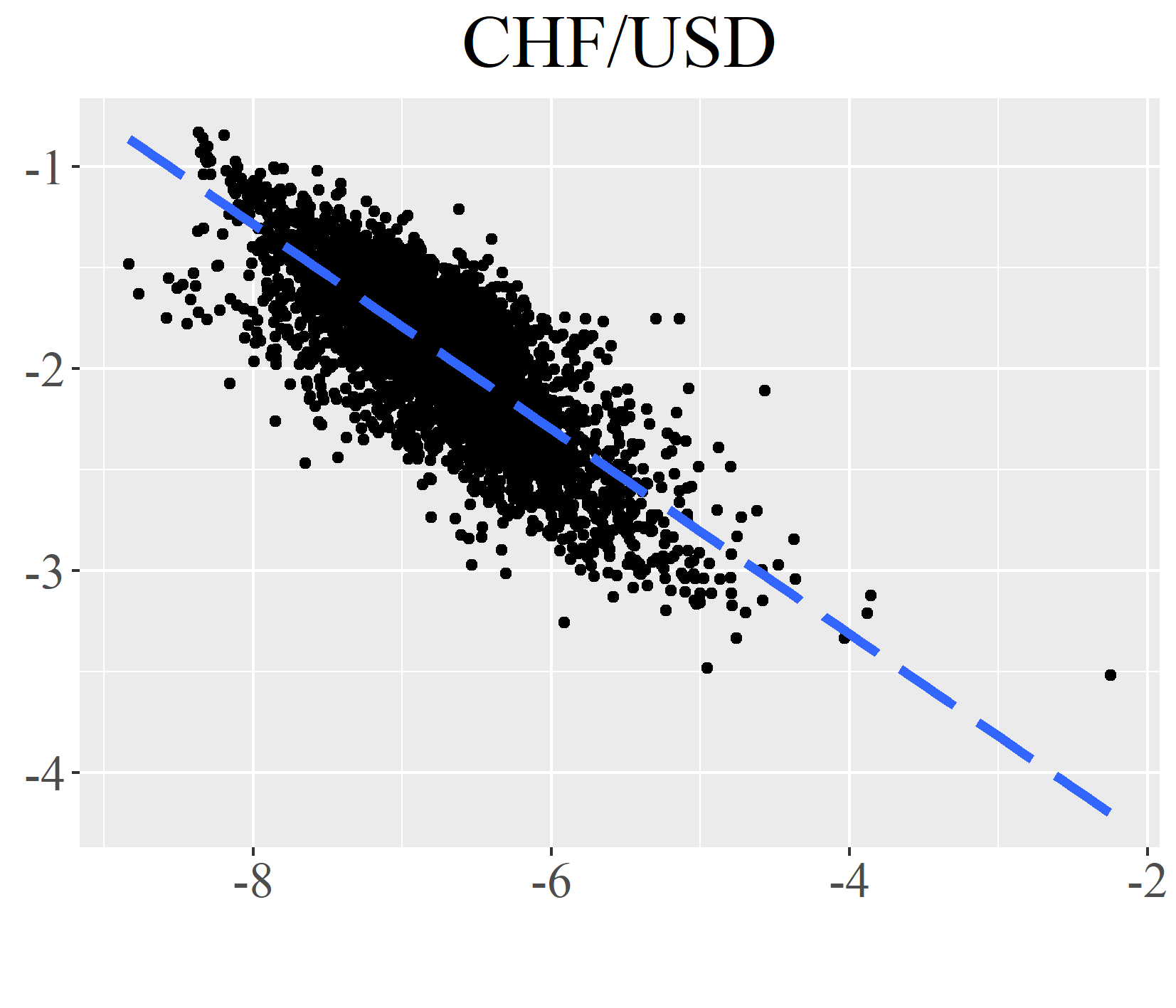}}
\caption{\label{figure:2}\footnotesize Display of $\ln\ell$ \emph{versus} $\ln \tfrac{\sigma}{\ell}$ (dashed lines) from data described in Table \ref{table:2}, where dots represent daily values.}
\end{figure}

\section{Concluding remarks} \label{sec:4}
\noindent

In this work, we have presented a general approach for understanding the relationship between the variance and the range of a general family of truncated distributions based on skewing functions. We have established a closed-form expression for its moments and their asymptotic behavior as the support's semi-range tends to zero and $\infty$.

As discussed previously by \cite{Matsushita2020, Matsushita2023}, if the truncated nature arises naturally from the past, the function relating truncation length and standard deviation may assist in connecting the bounded past and unbounded future data. For this reason, we expect our results to be useful in many practical situations.

%

\paragraph{Acknowledgements}
Roberto Vila gratefully acknowledges financial support from CNPq, CAPES, and FAP-DF, Brazil.
Raul Matsushita acknowledges financial support from CNPq, CAPES, FAP-DF, and DPI/DPG/UnB, Brazil. 
\paragraph{Disclosure statement}
There are no conflicts of interest to disclose.


\normalsize


\end{document}